\documentclass[10pt, conference, letterpaper]{IEEEtran}

\usepackage{cite}
\usepackage{amsmath,amssymb,amsfonts,amsthm}
\usepackage{algorithm}
\usepackage[noend]{algpseudocode}
\usepackage{graphicx}
\usepackage{textcomp}
\usepackage{xcolor}

\newtheorem{theorem}{Theorem}[section]
\newtheorem{lemma}[theorem]{Lemma}
\newtheorem{definition}[theorem]{Definition}

\newcommand{\note}[1]{
\bigskip
\noindent {\bf Note:} {\em #1}\\

\bigskip}

\usepackage{xifthen}

\newcommand{\bigO}[1]
{
    \ifthenelse{\isempty{#1}}
        {O}
        {O(#1)}
}

\newcommand{\edgeindnum}[1]
{
    \ifthenelse{\isempty{#1}}
        {\mathcal{V}}
        {\mathcal{V}(#1)}
}

\begin{document}

  \title{Random Gossip Processes in Smartphone Peer-to-Peer Networks}
 
\author{
  Newport, Calvin\\
  Georgetown University\\
  \texttt{cnewport@cs.georgetown.edu}
  \and
  Weaver, Alex\\
  Georgetown University\\
  \texttt{aweaver@cs.georgetown.edu}
}
  \maketitle

\begin{abstract}
In this paper, we study random gossip processes in communication models that describe the peer-to-peer networking functionality included in standard smartphone operating systems. Random gossip processes spread information through the basic mechanism of randomly selecting neighbors for connections.
These processes are well-understood in standard peer-to-peer network models,
but little is known about their behavior in models that abstract the smartphone peer-to-peer setting.
With this in mind, we begin by studying a simple random gossip process in the synchronous mobile telephone model (the most common abstraction used to study smartphone peer-to-peer systems).
By introducing a new analysis technique, we prove that this simple process is actually more efficient than the best-known gossip algorithm in the mobile telephone model, which required complicated
coordination among the nodes in the network.
We then introduce a novel variation of the mobile telephone model that removes the synchronized round assumption, shrinking the gap between theory and practice.
We prove that simple random gossip processes still converge in this setting and that information spreading still improves along with graph connectivity.
This new model and the tools we introduce provide a solid foundation for the further theoretical analysis of algorithms meant to be deployed on real smartphone peer-to-peer networks.
More generally, our results in this paper imply that simple random information spreading processes should be expected to perform well in this emerging new peer-to-peer setting.
%
%
%
\end{abstract}


\section{Introduction} \label{sec:intro}

In this paper, we study random gossip processes in smartphone peer-to-peer networks.
We prove the best-known gossip bound in the standard synchronous model used to describe this setting,
and then establish new results  in a novel asynchronous variation of this model that more directly matches the real world behavior of smartphone networks.
Our results imply that  simple information spreading strategies work surprisingly well in this complicated but increasingly relevant environment.

In more detail, a random gossip process is a classical strategy for spreading messages through a peer-to-peer network.
It has the communicating nodes randomly select connection partners from their eligible neighbors,
and then once connected exchange useful information.\footnote{The main place where different random gossip processes vary is in their definition of ``eligible." 
What unites them is the same underlying approach of random connections to nearby nodes.}
As elaborated in Section~\ref{sec:related},
these random processes are well-studied in standard peer-to-peer models where they have been shown to spread information efficiently
despite their simplicity.

To date, however, little is known about these processes in the emerging setting of {\em smartphone} peer-to-peer networks,
in which nearby smartphone devices connect with direct radio links that do not require WiFi or cellular infrastructure.
As also elaborated in Section~\ref{sec:related}, 
both Android and iOS now provide support for these direct peer-to-peer
connections, enabling the possibility of smartphone apps that generate large peer-to-peer networks
that can be deployed, for example, when  
 infrastructure is unavailable (i.e., due to a disaster) or censored (i.e., due to government repression).
 This paper investigates whether the random gossip processes that have been shown to spread information well in other peer-to-peer
 settings will prove similarly useful in this intriguing new context.
 
\subsection{The Mobile Telephone Model (MTM)}
The \emph{mobile telephone model} (MTM), introduced by Ghaffari and Newport \cite{ghaffari:2016}, extends the well-studied 
\emph{telephone model} of wired peer-to-peer networks (e.g.,\cite{frieze1985shortest,telephone1,telephone2,telephone3,giakkoupis2011tight,chierichetti2010rumour,giakkoupis2012rumor,fountoulakis2010rumor,giakkoupis2014tight})
to better capture the dynamics of standard smartphone peer-to-peer libraries.
In recent years, 
several important peer-to-peer problems have been studied in
the MTM, including rumor spreading \cite{ghaffari:2016}, load balancing \cite{dinitz:2017}, leader election \cite{newport:2017}, and gossip \cite{newport:2017b}.

As we elaborate in Section~\ref{sec:mtm},
the mobile telephone model describes a peer-to-peer network topology with an undirected graph,
where the nodes correspond to the wireless devices, and an edge between two nodes indicates 
the corresponding devices are close enough to enable a direct device-to-device link.
Time proceeds in synchronous rounds.
At the beginning of each round,
each node can {\em advertise} a bounded amount of information to its neighbors in the topology.
At this point, each node can then decide to either send a connection {\em invitation} to a neighbor,
or instead receive these invitations, choosing at most one incoming invitation to accept, forming a connection.
Once connected, a pair of node can perform a bounded amount of communication before the round ends.
Each node is limited to participate in at most one connection per round.

\subsection{Gossip in the MTM}
The gossip problem assumes that $k$ out of the $n \geq k$ nodes start with a gossip message. The problem is solved once all nodes have learned all $k$ messages.
In the context of the MTM, we typically assume that at most $O(1)$ gossip messages can be transferred over a connection in a single round, and that advertisements are bounded 
to at most $O(\log{n})$ bits.

A natural random gossip process in this setting is the following: In each round, each node advertises a hash of its token set and flips a fair coin to decide whether to send or receive connection invitations. If a node $u$ decides to send, and it has at least one neighbor advertising a different hash (implying non-equal token sets), then it selects among these neighbors with uniform randomness to choose a single recipient of a connection invitation. If the invitation is accepted the two nodes exchange a constant number of tokens in their set difference.

It is straightforward to establish that with high probability in $n$ this process solves gossip in $O(nk)$ rounds.
The key insight is that in every round there is at least one potentially productive connection, and that there can be at most $O(nk)$ such connections before all nodes know all messages.
(See~\cite{newport:2017b} for the details of this analysis.)

In our previous work on gossip in the MTM~\cite{newport:2017b},
we explored the conditions under which you could improve on this crude $O(nk)$ bound.
We conjectured that a more sophisticated analysis could show that simple random processes improve on the $nk$ bound given sufficient graph connectivity,
but were unable to make such an analysis work.
Accordingly, in~\cite{newport:2017b} we turned our attention to a more complicated gossip algorithm called {\em crowded bin}.
Unlike the simple structure of random gossip processes, crowded bin requires non-trivial coordination among nodes,
having them run a distributed size estimation protocol (based on a balls-in-bins analysis) on $k$, and then using these estimates
to parametrize a distributed TDMA protocol that eventually enables $k$ independent token spreading processes to run in parallel.

In~\cite{newport:2017b}, we prove that crowded bin solves gossip in $O((k/\alpha)\log^5{n})$ rounds, with high probability in $n$,
 when run in a network topology with vertex expansion $\alpha$ (see below).
For all but the smallest values of $\alpha$ (i.e., least amounts of connectivity), this result is an improvement over the crude $O(nk)$ bound achieved by the random process.

A key open question from this previous work is whether or not it is possible to close the time complexity gap between the appealingly simple random gossip
processes and the more complicated machinations of crowded bin.
As we detail next, this is the question tackled in this paper.

\subsection{New Result \#1:  Improved Analysis for Gossip in the MTM}
In Section~\ref{sec:mtm}, 
we consider a variation of the simple random gossip process described above modified only slightly such that a node only considers a neighbor eligible if 
it advertises a different hash {\em and} it has not recently attempted a connection with that particular neighbor. We call this variation {\em random spread} gossip

By introducing a new analysis technique, we significantly improve on the straightforward $O(nk)$ bound for random gossip processes like random spread.
Indeed, we prove this process is actually slightly {\em more} efficient than the more complicated crowded bin algorithm
from~\cite{newport:2017b}, showing that with high probability in $n$, random spread requires only $O((k/\alpha)\log^4{n})$ rounds to spread all the messages
in a network with vertex expansion $\alpha$.

The primary advantage of random spread gossip is its simplicity. As with most random gossip processes,
its behavior is straightforward and easy to implement as compared
to existing solutions. 
A secondary advantage is that this algorithm works in the {\em ongoing} communication scenario in which new rumors
keep arriving in the system. Starting from any point in an execution, if there are $k$ rumors that are not yet fully disseminated,
they will reach all nodes in at most an additional $O((k/\alpha)\log^4{n})$ rounds, 
regardless of how many rumors have been previously spread.
The solution in~\cite{newport:2017b},
by contrast,
must be restarted for each collection of $k$ rumors,
and includes no mechanism for devices to discover that gossip has completed for the current collection.
Accordingly, this new result fully supersedes the best known existing results for gossip in the MTM under similar assumptions.\footnote{In~\cite{newport:2017b},
we also study gossip under other assumptions, like changing communication graphs and the lack of good hash functions.}

At the core of our analysis is a new data structure we call a {\em size band table} that tracks
the progress of the spreading rumors.
We use this table to support an amortized analysis of spreading that proves that
the stages in which rumors spread slowly are balanced out sufficiently by stages in which they spread quickly,
providing a reasonable average  rate.

\subsection{New Result \#2: Gossiping in the Asynchronous MTM}
The mobile telephone model is a
high-level abstraction that captures the core dynamics of smartphone peer-to-peer communication,
but it does not exactly match the behavior of real smartphone networking libraries.
The core difference between theory and practice in this context is synchronization.
To support deep analysis, 
this abstract model (like many models used to study distributed graph algorithms)
synchronizes devices into well-defined rounds.
Real smartphones, by contrast, do not offer this synchronization.
It follows that algorithms developed in the mobile telephone model cannot be directly implemented on
real hardware.

With the goal of closing this gap,
in Section~\ref{sec:amtm}
we introduce the {\em asynchronous mobile telephone model} (aMTM),
a variation of the MTM that removes the synchronous round assumption,
allowing nodes and communication to operate at different speeds.
The main advantage of the aMTM, is that algorithms specified and analyzed
in the aMTM can be
directly implemented using existing smartphone peer-to-peer libraries.
The main disadvantage is that the introduction of asynchrony complicates analysis.

In Section~\ref{sec:amtm},
we first study the question of whether simple random gossip processes even still converge in a setting
where nodes and messages can operate at different speeds controlled by an adversary. 
We answer this question positively by proving that a simple random gossip process solves gossip in $O(nk\delta_{max})$
time, where $\delta_{max}$ is an upper bound on the maximum time certain key steps can occur
(as is standard, we assume $\delta_{max}$ is determined by an adversary, can change between executions, and is unknown to
the algorithm).

 We then tackle the question of whether it is still possible to show
 that the time complexity of
  information spreading improves with vertex expansion in an asynchronous setting.
  The corresponding analyses in the synchronous MTM,
  which treats nodes as implicitly running an approximate maximum matching algorithm between nodes that know a certain token and those that do not,
   depend heavily on the synchronization of node behavior.
   
   We introduce a novel analysis technique,
 in which we show that the probabilistic connection behavior in the aMTM over time sufficiently approximates
 synchronized behavior to allow our more abstract graph theory results to apply.
 In particular, we prove that for $k=1$, the single message spreads in at most $O(\sqrt{(n/\alpha)}\cdot \log^2{n\alpha} \cdot \delta_{max})$ time.
 This result falls somewhere between our previous $O(n \delta_{max})$ result for gossip with $k=1$ in the aMTM,
 and the bound of $O(\text{polylog}(n)/\alpha)$  rounds possible in the synchronous MTM for $k=1$.
 The remaining gap with the synchronous results seems due
  the ability of synchronous algorithms to keep a history of recent connection attempts (crucial to the underlying matching analysis),
 whereas in the asynchronous model such histories might be meaningless if some nodes are making connections attempts much faster than others.
 
We argue that our introduction of the aMTM, as well as a powerful set of tools for analyzing information spreading in this setting,
provides an important foundation for the future study of communication processes in realistic smartphone peer-to-peer models.

\section{Related Work}
\label{sec:related}

In recent years, there has been a growing amount of research on smartphone peer-to-peer 
networking~\cite{suzuki2012soscast,aloi2014spontaneous,reina2015survey,lu2016networking,holzer2016padoc,firechat,oghostpot} 
(see~\cite{nishiyama2014relay} for a survey).
%
There has also been recent work on using device-to-device links to improve cellular network performance, e.g., the inclusion of peer-to-peer connections
in the emerging LTE-advanced standard~\cite{doppler2009device,wang2015energy,liu2015device}, but these efforts differ from the peer-to-peer applications studied here as they typically assume coordination provided by the cellular infrastructure.

In this paper, 
we both study and extend the mobile telephone model  introduced in 2016 by Ghaffari and Newport~\cite{ghaffari:2016}.
This model modifies the classical telephone model of wired peer-to-peer networks (e.g., \cite{frieze1985shortest,telephone1,telephone2,telephone3,giakkoupis2011tight,chierichetti2010rumour,giakkoupis2012rumor,fountoulakis2010rumor,giakkoupis2014tight}) to better match the constraints and capabilities of the smartphone setting.
In particular, the mobile telephone model differs from the classical telephone model in that it allows small advertisements but restricts the number
of concurrent connections at a given node. As agued in~\cite{ghaffari:2016}, these differences (especially the latter) significantly change achievable results,
algorithm strategy, and
analysis techniques. The details of this model are inspired, in particular, but the multipeer connectivity framework offered in iOS.

Our random spread gossip algorithm disseminates $k$ rumors in at most $O((k/\alpha)\log^4{n})$ rounds in the mobile telephone model
in a network with $n$ nodes and vertex expansion $\alpha$ (see Section~\ref{sec:model:vertex}).
The previous best known algorithm for this model is crowded bin gossip~\cite{newport:2017b},
which is significantly more complicated and requires $O((k/\alpha)\log^5{n})$ rounds.\footnote{In~\cite{newport:2017b}, crowded bin
is listed as requiring $O((k/\alpha)\log^6{n})$ rounds, but that result assumes single bit advertisements in each round---requiring devices to spell out
control information over many rounds of advertising. To normalize with this paper,
in which tags can contain $\log{n}$ bits, crowded bin's time complexity improves by a $\log$ factor. We note that~\cite{newport:2017b} also explores slower
gossip solutions for more difficult network settings not considered here; e.g., changing network topologies and the absence of advertisements.}

To put these time bounds into context, 
we note that previous work in the mobile telephone model solved rumor spreading~\cite{ghaffari:2016}
 and leader election~\cite{newport:2017} in $O(\text{polylog}(n)/\alpha)$ rounds.
In the classical telephone model, 
a series of papers~\cite{chierichetti2010rumour,giakkoupis2012rumor,fountoulakis2010rumor,giakkoupis2014tight}  (each optimizing the previous) established that simple random rumor spreading requires $O(\log^2{n}/\alpha)$ rounds~\cite{giakkoupis2014tight},
which is optimal in the sense that for many $\alpha$ values,
there exists networks with a diameter in $\Omega(\log^2{n}/\alpha)$.
The fact that our gossip solution increases these bounds by a factor of $k$ (ignoring log factors) is natural
given that we allow only a constant number of tokens to be transferred per round.

As mentioned, random gossip processes more generally have been studied in other network models. These abstractions generally model time as synchronized rounds and by definition require nodes to select a neighbor uniformly at random in each round \cite{boyd2006gossip} \cite{karp2000rumor}. More recent work has demonstrated that these protocols take advantage of key graph properties such as vertex expansion and graph conductance \cite{moskaoyama2006conductance}. Asynchronous variants of these protocols have also been explored, where asynchrony is captured by assigning each node a clock following an unknown but well-defined probability distribution \cite{boyd2006gossip} \cite{giakkoupis2016rumor}.
The asynchronous MTM model introduced in our paper, by contrast, 
deploys a more general and classical approach to asynchrony in which an adversarial scheduler controls the time required for key events in a worst-case fashion.

\section{Random Gossip in the Mobile Telephone Model}
\label{sec:mtm}

Here we study a simple random gossip process in the mobile telephone model.
We begin by formalizing the model, the problem, and some graph theory preliminaries,
before continuing with the algorithm description and analysis.

\subsection{The Mobile Telephone Model}
\label{sec:mtm}

%
The mobile telelphone model describes a smartphone peer-to-peer network topology as an undirected connected graph 
$G=(V,E)$. A computational process (called a {\em node} in the following) is assigned to each vertex in $V$.
The edges in $E$ describe which node pairs are within communication range.
In the following, we use $u\in V$ to indicate both the vertex in the topology graph as well
as the computational process (node) assigned to that vertex.
We use $n=|V|$ to indicate the network size.

Executions proceed in synchronous rounds labeled $1,2,...$,
and we assume all nodes start during round $1$.
%
At the beginning of each round,
each node $u\in V$ selects an {\em advertisement} to broadcast to its neighbors $N(u)$ in $G$.
This advertisement is a bit string containing no more than $O(\log{n} + \ell_h)$ bits,
where $\ell_h$ is the digest length of a standard hash function parameterized to obtain the 
desired collision resistance guarantees.
After broadcasting its advertisement,
node $u$ then receives the advertisements broadcast by its neighbors in $G$ for this round.

At this point, $u$ decides to either {send} a {\em connection invitation}
to a neighbor, or passively receive these invitations.
If $u$ decides to receive, and at least one connection invitation arrives at $u$, 
then node $u$ can select at most one such incoming invitation to {\em accept}, forming a connection
between $u$ and the node $v$ that sent the accepted invitation.
Once $u$ and $v$ are connected,
they can perform a bounded amount of reliable interactive communication before the round ends,
where the magnitude of this bound is specified as a parameter of the problem studied.
Notice that the model does not guarantee to deliver $u$ all invitations sent to $u$ by its neighbors.
It instead only guarantees that if at least one neighbor of $u$ sends an invitation,
then $u$ will receive a non-empty subset (selected arbitrarily)
of these invitations before 
it must make its choice about acceptance. 

If $u$ instead chooses to send a connection invitation to a neighbor $v$,
there are two outcomes. If $v$ accepts $u$'s invitation,
a connection is formed as described above.
Otherwise, $u$'s invitation is implicitly rejected.

\subsection{The Gossip Problem}
\label{sec:model:gossip}
The gossip problem is parameterized with a token count $k>0$.
It assumes $k$ unique {\em tokens} are distributed to nodes at the beginning of the execution.
The problem is solved once all nodes have received all $k$ tokens.
We treat the tokens as black boxes objects that are large compared to the advertisements.
With this in mind, we assume the only ways 
for a node $u$ to learn token $t$ are: (1) $u$ starts with token $t$; or (2) a node $v$ that previously learned $t$
sends the token to $u$ during a round in which $v$ and $u$ are connected.

We assume that at most a constant number of tokens can be sent over a given connection.
Notice that this restriction enforces a trivial $\Omega(k)$ round lower bound for the problem.

\subsection{Vertex Expansion}
\label{sec:model:vertex}

Some network topologies are more suitable for information dissemination than others.
In a clique, for example, a message can spread quickly through epidemic replication,
while spreading a message from one endpoint of a line to another is necessarily slow.
With this in mind,
the time complexity of
information dissemination algorithms are often expressed with respect to graph connectivity
metrics such as {\em vertex expansion} or {\em graph conductance}.
In this way, an algorithm's performance can be proved to improve along with available connectivity.

In this paper, 
as in  previous studies of algorithms in the mobile telephone model~\cite{ghaffari:2016,dinitz:2017,newport:2017,newport:2017b},
we express our results with respect to vertex expansion (see~\cite{ghaffari:2016} for an extended discussion
of why this metric is more appropriate than conductance in our setting).
Here we define this metric and establish a useful related property.

For fixed undirected connected graph $G=(V,E)$,
and a given $S \subseteq V$, we define the {\em boundary} of $S$, indicated $\partial S$, as follows:
 $\partial S = \{ v\in V \setminus S : N(v) \cap S \neq \emptyset\}$: that is, $\partial S$ is the set
 of nodes not in $S$ that are directly connected to $S$ by an edge in $E$.
 We define $\alpha(S) = |\partial S|/|S|$.
We define the {\em vertex expansion} $\alpha$ of a given graph $G = (V,E)$
 as follows:
 
 \[  \alpha = \min_{S \subset V, 0 < |S| \leq n/2} \alpha(S). \]
 
 \noindent Notice that despite the possibility of $\alpha(S) >1$ for some $S$, we always have $\alpha \leq 1$.
In more detail, this parameter ranges from $2/n$ for poorly connected graphs (e.g., a line)
to values as large as $1$ for well-connected graphs (e.g., a clique).
Larger values indicate more potential for fast information dissemination.

The mobile telephone model requires the set of pairwise connections in a given round to form a matching
in the topology graph $G=(V,E)$.
The induces a connection between
maximum matchings and the maximum amount of potential communication in a given round.
Here we adapt a useful result from~\cite{ghaffari:2016} that formalizes
the relationship between vertex expansion and these matchings as defined with respect
to given partition.

In more detail, 
for a given graph $G=(V,E)$ and node subset $S\subset V$,
we define $B(S)$ to be the bipartite graph with bipartitions $(S,V \setminus S)$,
and the edge set $E_S = \{ (u,v): (u,v) \in E$, $u\in S$, and $v\in V \setminus S\}$.
 Recall that the {\em edge independence number} of a graph $H$,
denoted $\nu(H)$, describes the size of a maximum matching on $H$.
For a given $S$, therefore, $\nu(B(S))$ describes the maximum number 
of concurrent connections that a network can support in the mobile telephone model between nodes in $S$
and nodes outside of $S$.  This property follows from the restriction in this model that each node can participate
in at most one connection per round. 

The following result notes that the vertex expansion does a good job of approximating the size
of the maximum matching across any partition:

 \begin{lemma}[from~\cite{ghaffari:2016}]
 Fix a graph $G=(V,E)$ with $|V| = n$ with vertex expansion $\alpha$.
Let $\gamma = \min_{S\subset V, |S| \leq n/2}\{ \nu(B(S))/|S|  \}$.
It follows that $\gamma \geq \alpha/4$.
\label{lem:msize}
\end{lemma}



\subsection{The Random Spread Gossip Algorithm}
\label{sec:alg}

 
We formalize our random spread gossip algorithm with the pseudocode labeled Algorithm~1.
Here we summarize its behavior.

The basic idea of the algorithm is that in each round, each node advertises a hash of their token set.
Nodes then attempt to connect only to neighbors that advertised a different hash, 
indicating their token sets are different. 
When two nodes connect,
they can transfer a constant number of tokens in the non-empty set difference
of their respective token sets.

As detailed in the pseudocode, the random spread algorithm implements the above strategy combined with some minor additional structure
that supports the analysis. 
In particular, nodes partition rounds into {\em phases} of length $\lceil \log{N} \rceil$, where $N > 1$ is an upper bound on the maximum
degree $\Delta$ in the network topology.
Instead of each node deciding whether to send or receive connection invitations at the beginning of {\em each round},
they make this decision at the beginning of {\em each phase}, and then preserve this decision throughout the phase (this is captured in the pseudocode
with the $status$ flag that is randomly set every $\lceil \log{N} \rceil$ rounds).
Each receiver node also advertises whether or not it has been involved in a connection already during the current phase
(as captured with the $done$ flag).
A sender node will only consider neighbors that advertise a different hash, are receivers in the current phase,
{\em and} have not yet been involved in a connection during the phase.

\begin{algorithm}

\begin{algorithmic}
\State
\State \underline{Initialization:}  
\State $N \gets$ upper bound on maximum degree in topology
\State $T \gets$ initial tokens (if any) known by $u$
\State $H$ is a hash function
\State
\State \underline{For each round $r$:}
\State
\If{$r$ mod $\lceil \log{N} \rceil$ = 1}
	\State $status \gets$ random bit (1=sender; 0=receiver)
	\State $done \gets 0$
\EndIf
\State
\State {\tt Advertise}$(\langle status, done, H(T,r), u \rangle)$ 
\State $A \gets$ {\tt RecvAdvertisements}$()$
\State
\State $A' \gets \{ v\mid \langle 0,0,h,v \rangle \in A, h \neq H(T,r)\}$
\If{$status=1$ {\bf and} $|A'| > 0$} 
	\State $v \gets$ node selected with uniform randomness from $A'$
	\State {\em (Attempt to connect with $v$. If successful, exchange a token in set difference.)}
\ElsIf{$status=0$} 
	\State {\em (If receive connection proposal(s): accept one, exchange token in the set difference,  set $done \gets 1$.)}
\EndIf
\end{algorithmic}

\label{alg:1}
\caption{Random spread gossip (for node $u$).}
\end{algorithm}

\subsection{Analysis of Random Spread Gossip}
Our goal is to prove the following result about the performance of random spread gossip:

\begin{theorem}
With high probability, 
the random spread gossip algorithm solves the gossip problem 
in $O((k/\alpha)\log^2{n}\log{N}\log{\Delta})$ rounds, when executed with $k>0$ initial tokens and degree bound $N\geq \Delta$,
 in a network topology graph of size $n$, maximum degree $\Delta$, and vertex expansion $\alpha$.
\label{thm:main}
\end{theorem}

We begin by establishing some preliminary notations and assumptions before continuing to
the main proof argument.

\paragraph{Notation}
For a fixed execution, 
let $Q$ be the non-empty set of $k$ tokens that the algorithm must spread.
For each round $r>0$ and node $u\in V$,
let $T_u(r)$ be the tokens (if any) ``known" by $u$ at the start of round $r$
(that is, the tokens that $u$ starts with as well as every token it received through a connection in rounds $1$ to $r-1$).

For each $t\in Q$, and round $r>0$, let $S_t(r) = \{  v: t\in T_v(r)\}$ be the nodes that know token $t$ at the start of round $r$.
Let $n_t(r) = |S_t(r)|$ be the number of nodes that know token $t$ iat the beginning of this round,
and let $n^*_t(r) = \min\{ n_t(r), n-n_t(r)   \}$.

Finally, let  
 $t^*(r)=\text{argmax}_{t\in Q}\{n^*_t(r)\}$ be a token $t$ with the maximum $n^*_t(r)$ value in this round (breaking ties arbitrarily).
 According to Lemma~\ref{lem:msize},
which connects vertex expansion to matchings,
there is a matching between nodes in $S_{t^*(r)}(r)$ and $V\setminus S_{t^*(r)}(r)$ of
size at least $(\alpha/4)\cdot n^*_{t^*(r)} (r)$.
Token $t^*(r)$, in other words,
has the largest guaranteed potential to spread in round $r$ among all tokens.\footnote{To be slightly more precise,
$(\alpha/4)\cdot n^*_{t^*(r)}(r)$ is a lower bound on the size of the matching across the cut
defined by $t^*(r)$, so $t^*(r)$ is the token with the largest lower bound guarantee on the size of its matching.}
Accordingly, in the analysis that follows,
we will focus on this token in each phase to help lower bound the amount of spreading we hope to achieve.

 \paragraph{Productive Connections and Hash Collisions}
 In the following, we say a given pairwise connection between nodes $u$ and $v$  in some round $r$ is {\em productive} if 
 $T_u(r) \neq T_v(r)$. That is, at least one of these two nodes learns a new token during the connection.
 By the definition of our algorithm, 
 if $u$ and $v$ connect in round $r$,
 then it must be the case that $H(T_u(r),r) \neq H(T_v(r),r)$,
 where $H$ is the hash function used by the random spread gossip algorithm.
 This implies $T_u(r) \neq T_v(r)$---indicating that every connection created by our algorithm is productive.
 
 On the other hand, it is possible for some $u$, $v$, and $r$ that even though $T_u(r) \neq T_v(r)$,
 $H(T_u(r),r) = H(T_v(r),r)$ due to a hash collision.
 For the sake of clarity, in the analysis that follows we assume that no hash collisions occur in the analyzed execution.
 Given the execution length is polynomial in the network size $n$,
 and there are at most $n$ different token sets hashed in each round,
 for standard parameters the probability of a collision among this set would be extremely small, supporting our assumption.
 
 We emphasize, however, that even if a small number of collisions {\em do} occur, 
 their impact is minimal on the performance of random spread gossip.
 The worst outcome of a hash collision in a given round
 is that during that single round a potentially productive connection is not observed to be productive and therefore 
 temporarily ignored. As will be made clear in the analysis that follows, the impact of this event is nominal.
 Indeed,
 even if we assumed that up to a constant fraction of the hashes in every round generated collisions---an extremely unlikely event for all but the weakest hash function parameters---the algorithm's worst case time complexity would decrease
 by at most a constant factor.
  
\paragraph{Matching Phases}
Recall that our algorithm partitions rounds into {\em phases} of length $\lceil \log{N}\rceil$.
For each phase $i>0$,
let $r_i = \lceil \log{N} \rceil\cdot (i-1) + 1$ be the first round of that phase.
Fix some arbitrary phase $i$
and consider token $t=t^*(r_i)$,
which, as argued above, is the token with the largest guaranteed potential to spread
in round $r_i$.
Our goal in this part of the analysis is to prove that with constant probability,
our algorithm will create enough productive connections during this phase to well-approximate
this potential.
This alone is not enough to prove our algorithm terminates efficiently,
as in some phases, it might be the case that {\em no} token has a large potential to spread.
The next part of our argument will tackle this challenge by proving that over a sufficient number
of phases the aggregate amount of progress must be large.

We begin by
establishing the notion of a {\em productive subgraph}:

\begin{definition}
At the beginning of any round $r>0$, we define the {\em productive subgraph} of the network
topology $G=(V,E)$ for $r$ as: $G_r = (V,E_r)$, where
$E_r = \{ \{u,v\}\mid \{u,v\}\in E, T_u(r) \neq T_v(r),u.status(r) \neq v.status(r)\}$,
and for each $w\in V$, $w.status(r)$ indicates the value of the node $w$'s status bit 
for the phase containing round $r$. 
\end{definition}

That is, the productive subgraph for round $r$ is the subgraph of $G$ that contains
only edges where the endpoints: (1) have different token set; and (2) have
different statuses (one is a sender during this phase and one is a receiver).
This subgraph contains every possible connection for a given round of our gossip algorithm
(we ignore $done$ flags because, as will soon be made clear,
we consider these graphs defined only for the first round of phases, a point at
which all $done$ flags are reset to $0$).
Accordingly, a maximum matching on this subgraph upper bounds the maximum number
of concurrent connections possible in a round.

We begin by lower bounding the size of the maximum matching in a productive subgraph at the beginning of a given phase $i$
using the token $t=t^*(r_i)$. Recall that $n^*_t(r_i)$ is the number of nodes that know token 
$t$ at the beginning of $r$, if less than half know the token, and otherwise indicates
the number of nodes that do not know $t$.

\begin{lemma}
Fix some phase $i$.
Let $t= t^*(r_i)$.
Let $G_{r_i}$ be the productive subgraph for round $r_i$,
 $M_{i}$ be a maximum matching on $G_{r_i}$, and $m_i = |M_{i}|$.
With constant probability (defined over the $status$ assignments):
$m_i \geq (\alpha/16)n^*_t(r_i)$. 
\label{lem:prod}
\end{lemma}
\begin{proof}
Fix some phase $i$.
We define $G'_{r_i}$ to be the {\em potentially productive} subgraph for round $r_i$,
where {\em potentially productive} is defined the same as {\em productive} except
we {\em omit} the requirement that endpoints of edges in the graph have different $status$ values.
 Let $M'$ be a maximum matching on $G'_{r_i}$ and $m' = |M'|$.
 We will reason about $m'$ as an intermediate step toward bounding the size of the actual productive subgraph
 for this round.

Let $t=t^*(r_i)$.
Consider the cut between nodes that know $t$, and nodes that do not,
 at the beginning of this phase.
By Lemma~\ref{lem:msize}, 
there is a matching across this cut of size at least $(\alpha/4)n^*_t(r_i)$.
By definition, for all edges across this cut, their endpoints have different
token sets at the beginning of round $r_i$, 
therefore they are all candidates to be included in $M'$,
implying that $m' \geq  (\alpha/4)n^*_t(r_i)$.

Our next step is to consider the random assignment of sender and receiver status to nodes in $M'$
at the beginning of phase $i$.
For an edge in $M'$ to be included in a matching on the productive subgraph $G_{r_i}$,
 it must be the case that 
one endpoint chooses to be a receiver while the other chooses to be a sender.
We call such an edge {\em good}.
For any particular edge $e \in M'$,
this occurs with probability $1/2$. 

For each such $e\in M'$, 
 let $X_e$ be the random indicator that evaluates to $1$ if $e$ is good, and evaluates to $0$ otherwise.
 Let $Y=\sum_{e\in M'} X_e$ be the number of good edges for this phase.
By our above probability calculation, we know:

\[ E[Y] = E\left[ \sum_{e\in M'} X_e \right] = \sum_{e\in M'} E[X_e] = m'/2. \]

Because $M'$ is a matching, these indicator variables are independent. 
This allows us to concentrate on the mean.
In particular,
we will apply the following multiplicative Chernoff Bound, 
defined for $\mu = E[Y]$ and any $0\leq \delta \leq 1$:

\[ \Pr(Y \leq (1-\delta)\mu) \leq e^{- \frac{\delta^2\mu}{2}}, \]

\noindent with $\delta=1/2$, to establish that the probability that $Y \leq m'/4$ is upper bounded by:

\[  e^{- \frac{\mu}{8}} = e^{- m'/16   } < .94. \]

It follows that $Y$ is less than or equal to  $m'/4$, which is itself greater
than or equal to $(\alpha/16)n^*_t(r_i)$
with a probability upper bounded by a constant---as required.\footnote{Clearly,
 the specific worst failure bound of $0.06$ is loose (in the worst case, where $m_i=1$, for example,
we can directly calculate that $Y=m_i$ with probability $1/2$). We are not,
however, attempting to optimize constants in this analysis,
so any constant bound is sufficient
for our purposes.}
\end{proof}

We now turn our attention to our gossip algorithm's ability to take advantage
of the potential productive connections captured by the productive subgraph defined
at the beginning of the phase. 
To do so, we first adapt a useful result on rumor spreading from~\cite{ghaffari:2016}
to the behavior of our gossip algorithm.
Notice that it is the proof of the below adapted lemma that requires
the use of the $done$ flag in our algorithm.

\begin{lemma}[adapted from Theorem 7.2 in~\cite{ghaffari:2016}]
Fix a phase $i$.
Let $G'$ be a subgraph of the productive subgraph $G_{r_i}$ that satisfies the following:
\begin{enumerate}
\item there is a matching of size $m$ in $G'$;
\item the set $L$ of nodes in $G'$ with sender status is of size $m$; and
\item for each node $u\in L$, every neighbor of $u$ in $G_{r_i}$ is in $G'$. 
\end{enumerate}
With constant probability (defined over the random neighbor choices), 
during the first $\log{\Delta}$ rounds of phase $i$,
at least $\Omega\left(\frac{m}{\log{n}\log{\Delta}}\right)$ neighbors of nodes in $L$ in $G'$
participate in a productive connection.
\label{lem:match} 
\end{lemma}
 \begin{proof}[Proof Notes]
The original version of this theorem from~\cite{ghaffari:2016} 
requires that $G'$ is a bipartite graph. 
This follows in our case because it is a subset of a productive subgraph.
All subsets of productive subgraphs are bipartite as you can put the nodes
with sender status in one bipartition and nodes with receiver status in the other
(by definition the only edges in a productive subgraph are between sender and receiver nodes).

Another difference is that our theorem studies our gossip
algorithm, while the theorem from~\cite{ghaffari:2016} studies the PPUSH rumor spreading process.
The PPUSH process assumes  a single rumor spreading in the system.
Some nodes know the rumor (and are called {\em informed}) and some nodes do not (and are called {\em uninformed}).
In each round, each node declares whether or not they are uninformed.
Each informed node randomly chooses an uninformed neighbor (if any such neighbors exist)
and tries to form a connection, changing the receiver's status to informed.

The original version of the theorem states that if you execute PPUSH for $\log{\Delta}$
rounds, at least $\Omega(\frac{m}{\log{n}\log{\Delta}})$ nodes that neighbor $L$ are informed.
If we consider senders to be informed and receivers to be uninformed, 
our gossip algorithm behaves the same as PPUSH in $\log{\Delta}$ rounds under consideration.
That is, the senders in $L$ will randomly select a receiver neighbor to attempt a connection.

Once a receiver in $G'$ participates in a connection in our algorithm,
it sets its $done$ flag to $1$ for the remainder of the phase,
preventing future attempts to connect to it during the phase.
This matches the behavior in PPUSH where once a node becomes informed,
informed neighbors stop trying to connect to it.
This congruence allows us to derive the same $\Omega(\frac{m}{\log{n}\log{\Delta}})$ bound
derived for PPUSH in~\cite{ghaffari:2016}.
 \end{proof}
 
 We now combine Lemmas~\ref{lem:prod} and~\ref{lem:match}
 to derive our main result for this part of the analysis.

\begin{lemma}
Fix some phase $i$.
Let $t= t^*(r_i)$.
With constant probability,
the number of productive connections in this phase is in
$\Omega\left(\frac{\alpha n^*_t(r_i)}{\log{n}\log{\Delta}}\right)$.
\label{lem:progress}
\end{lemma}
 \begin{proof}
 Fix some phase $i$.
 By Lemma~\ref{lem:prod},
 with some constant probability $p_1$,
 the productive subgraph $G_{r_i}$ has a matching
 $M_i$ of size $m_i \geq (\alpha/16)n^*_t(r_i)$
 once nodes randomly set their $status$ flags.
 
 Now consider the subgraph graph $G'$ that consists of every sender
 endpoint in $M_i$,
 and for each such sender $u$,
 every receiver $v$ that neighbors $u$,
 as well as the edge $\{u,v\}$.
 This subgraph satisfies the conditions of Lemma~\ref{lem:match} for $m=m_i$.
 Applying this lemma, it follows that with some constant probability $p_2$,
 during this phase, the random neighbor selections by senders
 will generate at least $\Omega(m_i/(\log{n}\log{\Delta}))$
 productive connections.
 
 Combining these two results,
 we see that with constant probability $p=p_1p_2$,
 we have at least $\Omega(\frac{\alpha n^*_t(r_i)}{\log{n}\log{\Delta}})$ productive connections, 
 as claimed by the lemma statement.
 \end{proof}

\paragraph{The Size Band Table}
In the previous part of this analysis, we
proved that with constant probability the number of productive connections in phase
$i$ is bounded with respect to the number of nodes that know $t^*(r_i)$.
In the worst case, however, $t^*(r_i)$ might be quite small (e.g., at the beginning
of an execution where each token is known by only a constant number of nodes,
this value is constant). We must, therefore, move beyond a worst-case application
of Lemma~\ref{lem:progress}, and amortize the progress over time to something more
substantial.

To accomplish this goal, we introduce a data structure---a tool used only in the context of our analysis---that we
call a {\em size band table}, which we denote as ${\cal S}$.
This table has one column for each token $t\in T$,
and $2\log{(n/2)}+1$ rows which we number $1,2,...,2\log{n/2}+1$.

As we will elaborate below, each row is associated with a range of values that
we call a {\em band}.
We call rows $1$ through $\log{(n/2)}$  {\em growth bands},
and rows $\log{(n/2)} + 1$ through $2\log{(n/2)}+1$ {\em shrink bands}.
Each cell in ${\cal S}$  contains a single bit.
We update these bit values after every round of our gossip algorithm
to reflect the extent to which each token has spread in the system.

In more detail, for each round $r \geq 1$,
we use ${\cal S}_r$ to describe the size band table at the beginning of round $r$.
For each token $t\in T$ and row $i, 1 \leq 1 \leq 2\log{(n/2)}+1$,
we use ${\cal S}_r[t,i]$ to refer to the bit value in row $i$
of the column dedicated to token $t$ in the table for round $r$.

Finally, we define each of these bit values as follows.
For each round $r \geq 1$,
token $t\in T$, and
growth band $i$ (i.e., for each $i, 1 \leq i \leq \log{(n/2)}$),
we define:

\[
{\cal S}_r[t,i] = \begin{cases}
			1 & \text{if at least $2^i$ nodes know}\\ 
			    & \text{token $t$ at the beginning of round $r$,} \\
			0 & \text{else.}
			\end{cases}
\]

Symmetrically,
for each round  $r \geq 1$,
token $t\in T$, and
shrink band $i$ (i.e., for each $i, \log{(n/2)} + 1 \leq i \leq 2\log{(n/2)} + 1$),
we define:

\[
{\cal S}_r[t,i] = \begin{cases}
			1 & \text{if less than $\frac{n}{2^{i-\log{(n/2)}}}$ nodes do {\em not}}\\
			    & \text{know token $t$ at the beginning of round $r$,} \\
			0 & \text{else.}
			\end{cases}
\]

A key property of the side band table is that as a given token $t$ spreads,
the cells in its column with $1$ bits grow from the smaller rows toward the larger rows. 
That is, if row $i$ is $1$ at the beginning of a given round,
all smaller rows for that token are also $1$ at the beginning of that round.
Furthermore, because nodes never lose knowledge of a token, once a cell is set
to $1$, it remains $1$.

\begin{figure}[htbp]
\centerline{\includegraphics[width=0.45\textwidth]{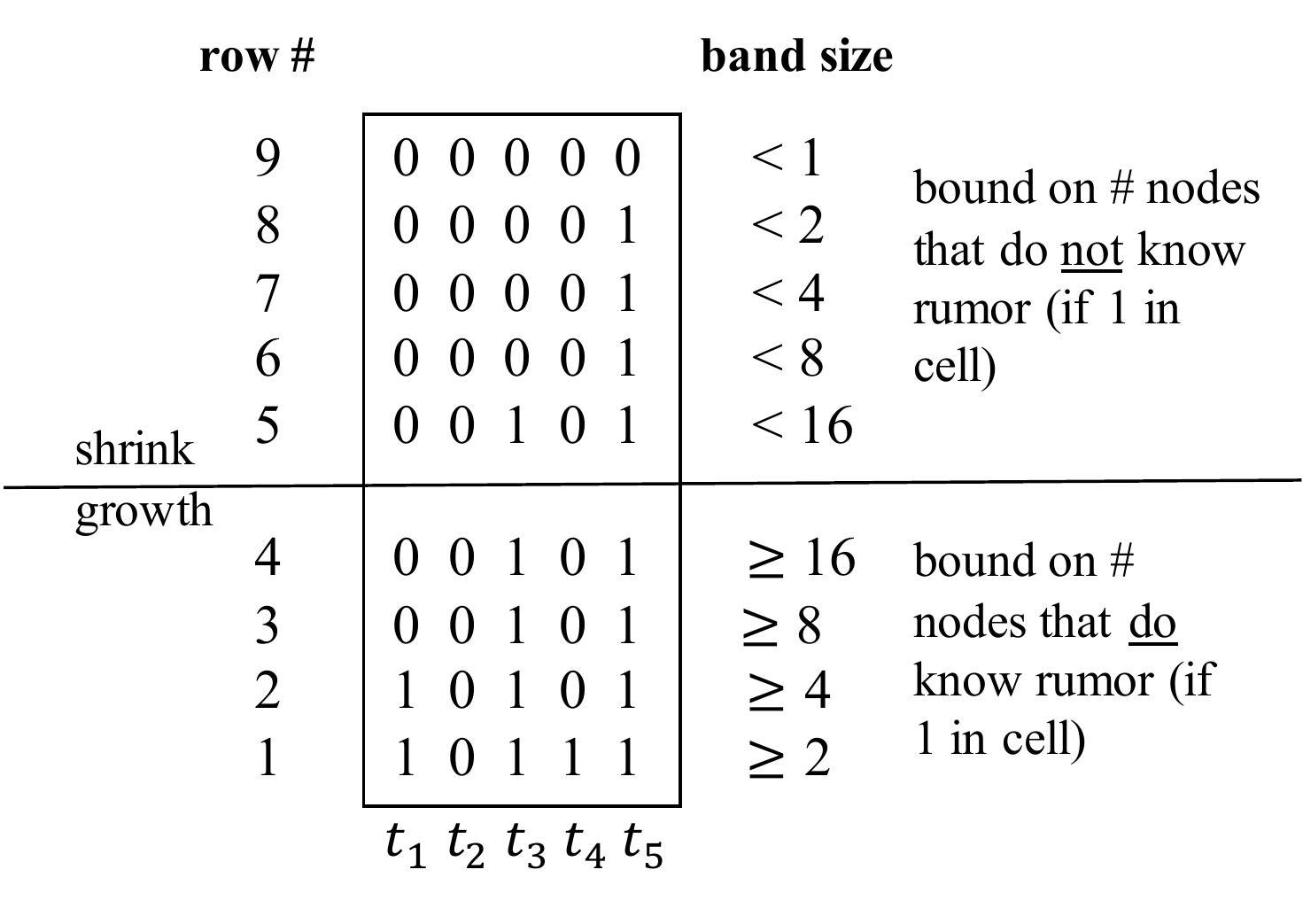}}
\caption{An example size band table for token set $T=\{t_1,t_2,t_3,t_4,t_5\}$ and network size $n=32$.
There is one column for each token. The largest row containing a $1$ for a given token
bounds the token spread.
Token $t_1$, for example, has spread to at least $4$ out of the $32$ nodes,
while token $t_5$ is known to all but $1$ node (indicating that it has spread to at least $31$).
In this example table, token $t_3$, which is spread to somewhere between $16$ to $24$ nodes,
 has the biggest potential to spread in the current round}
\label{fig}
\end{figure}

When all rows for a given token $t$ are set to $1$, it follows that
all nodes know $t$.
This follows because the definition of shrink band $i=2\log{(n/2)}+1$ being set to $1$ 
is that the number of nodes that do {\em not} know $t$ is strictly {\em less}
than:

\begin{eqnarray*}
\frac{n}{2^{i-\log{(n/2)}}} & = & \frac{n}{2^{2\log{(n/2)} + 1 -\log{(n/2)}}}  \\
 &=&  
\frac{n}{2^{\log{(n/2)} + 1}} \\
&=& \frac{n}{2^{\log{(n/2)}} \cdot 2^1} \\
&=& 1.
\end{eqnarray*}

\paragraph{Amortized Analysis of Size Band Table Progress}
As the size band table increases the number of $1$ bits,
we say it {\em progresses} toward a final state of all $1$ bits.
Here we perform an amortized analysis of size band table progress.

To do so, we introduce some notation.
For each phase $i$,
and token $t\in T$,
let $b_t(i)$ be the largest row number that contains a $1$ in $t$'s column in ${\cal S}_{r_i}$.
We call this the {\em current band} for token $t$ in phase $i$.

Let $a(i) = |b_{t^*(r_i)}(i) - \log{(n/2)}|$ define the distance from 
the current band of token $t^*(r_i)$ to the center row number $\log{(n/2)}$.
By the definition of $t^*(r_i)$,
no token has a current band closer to $\log{(n/2)}$ than $t^*(r_i)$ at the start of phase $i$.
We say that phase $i$ is {\em associated} with the current band for $t^*(r_i)$.

Finally, for a given phase $i$, with $t=t^*(r_i)$,
we say this phase is {\em successful} if the number of productive connections
during the phase is at least as large as the lower bound specified by Lemma~\ref{lem:progress};
i.e., there are at least $\frac{\gamma \alpha n^*_{t}(r_i)}{\log{n}\log{\Delta}}$ productive connections.
where $\gamma>0$ is the constant hidden in the asymptotic bound in the lemma statement.

%

Our first goal in this part of the analysis,
 is to bound the number of successful phases that can be associated with each band.
To do so, we differentiate between two different types of successful phases,
and then bound each separately.

\begin{definition}
Fix some phase $i$ that is associated with some band $j$ at distance $a(i)$ from the center of the size band table.
We say phase $i$ is an {\em upgrade} phase if there exists a subset of the productive connections during phase $i$
that push some token $t$'s current band to a position $j'$ with $|\log{(n/2)} - j'| < a(i)$.
If a phase is not an upgrade phase,
and at least one node is missing at least one token, we call it a {\em fill} phase.
\end{definition}

Stated less formally, we call a phase an upgrade phase if it pushes some token's count closer to the center of the size band
table---row $\log{(n/2)}$---than the band associated with the phase. Our definition is somewhat subtle in that it must handle the case where during a phase a token count
does grow to be closer to the center of the size band table, but then its count continues to grow until it pushes {\em more than} distance $a(i)$ above
the center. We still want to count this as an upgrade phase (hence the terminology about there existing some {\em subset} of the connections
that push the count closer).

Our goal is to bound the number of successful phases possible before all tokens are spread.
We begin with bound on upgrade phases (which hold whether or not
the phase is successful). Our subsequent bound on fill phases, however, considers only successful phases.

\begin{lemma}
There can be at most $k(2\log{(n/2)} + 1)$ upgrade phases.
\label{lem:phase1}
\end{lemma}
\begin{proof}
Fix some band $j$.
Consider an upgrade phase $i$ that is associated with $j$.
By the definition of an upgrade phase, there is some token $t$ with a current band at the start of $i$
that is distance at least $a(i)$ from the center of table, but that has its count grow closer to the center
during the phase. 

We note that it must be the case that $t$'s current band at the start of phase $i$ is a growth band.
This holds because if $t$'s current band is a shrink band then additional spreading of token $t$ can
only {\em increase} its distance from the center of the size band table.

If $j$ is a growth band, then it follows that $t$'s current band starts phase $i$ no larger than $j$ and
ends phase $i$ larger, because current bands for a token never decrease.
Moving forward, therefore, token $t$ can never again be the cause of a phase associated with band $j$ to be categorized as an upgrade phase.

On the other hand, if $j$ is a shrink band,
we know that after phase $i$, 
token $t$'s distance will remain closer to the center of the table than $j$ until $t$'s current band becomes a shrink band.
Once again, therefore, moving forward token $t$ can never again cause a phase associated with band $j$ to be categorized as an upgrade.

The lemma statement follows as
there are  $2\log{(n/2)} + 1$ bands,
and for each band, each of the $k$ tokens can transform that band into an upgrade phase at most once.
\end{proof}

We now bound the number of successful fill phases.
To do so, we note that the number of fill phases associated with a given band is bounded by the
worst case number of connections needed before some token's count must advance past that band.
For bands associated with large ranges this worst case number is large.
As shown in the following lemma, however, the number of connections in phases
associated with large bands grows proportionally large as well.
This balancing of growth required and growth obtained is at the core of our amortized analysis.

\begin{lemma}
There can be at most $O((k/\alpha)\log^2{n}\log{\Delta} )$ successful fill phases.
\label{lem:phase2}  
\end{lemma}
\begin{proof}
Consider a group of successful fill phases associated with some band $j$ at distance $a_j$ from the center of the size band table.
Because these are fill phases, the productive connections generated during these phases can never
push some token's count (perhaps temporarily) closer than distance $a_j$ from the center of the table
(any phase in which this occurs becomes, by definition, an upgrade phase).

One way to analyze the distribution of the productive connections during these phases is to consider
a generalization of the size band table in which we record in each cell $[t,i]$ the total number of
productive connections that spread token $t$ while its count falls into the band associated with row $i$.
(Of course, many connections for a given token might occur in a given round, in which we case, we process
them one by one in an arbitrary order while updating the cell counts.) 

If we apply this analysis only for the fill phases fixed above, then we know that the counts in all cells of distance less than
$a_j$ from the center of the table remain at $0$.
By the definition of the size band table, 
for a given token $t$,
the maximum number of connections we can add to cells of distance at least $a_j$ from the center
is loosely upper bounded by $2\cdot 2^{\log{(n/2)} - a_j}$ (the extra factor of two captures
both growth and shrink band cells at least distance $a_j$).
Therefore,
the total number of productive connections we can process into cells at distance at least $a_j$
is at most $2k2^{\log{(n/2)} - a_j}$.

By the definition,
each phase $i$ that is a successful fill phase associated with $j$ generates at least $\frac{\gamma \alpha n^*_{t}(r_i)}{\log{n}\log{\Delta}}$
productive connections, where $t=t^*(r_i)$.
By the definition of $t^*(r_i)$, $t$'s current band is distance $a_j$ from the center.
Therefore, $n^*_{t}(r_i)$ is within a factor of $2$ of $2^{\log{(n/2)} - a_j}$.
By absorbing that constant factor into the constant $\gamma$ (to produce a new constant $\gamma'$),
it follows that this phase generates at least 

\[ z = \frac{\gamma' \alpha 2^{\log{(n/2)} - a_j}}{\log{n}\log{\Delta}}\] 

\noindent new productive connections. Combined with our above upper bound on the total possible productive connections for successful
fill phases associated with $j$, it follows that the total number of successful fill phases associated with $j$ as less than:

\begin{eqnarray*}
z^{-1}2k2^{\log{(n/2)} - a_j} &=&  \left(\frac{\log{n}\log{\Delta}}{\gamma' \alpha 2^{\log{(n/2)} - a_j}}\right) 2k2^{\log{(n/2)} - a_j} \\
 &= & \Theta((k/\alpha)\log{n}\log{\Delta}).
\end{eqnarray*}

We multiply this bound over $2\log{(n/2)} + 1$ possible bands 
to derive $O((k/\alpha)\log^2{n}\log{\Delta})$ total
possible successful fill phases,
providing the bound claimed by the lemma statement.
\end{proof}

\paragraph{Pulling Together the Pieces}
We are now ready to combine the above lemmas to prove our main theorem.

\begin{proof}[Proof (of Theorem~\ref{thm:main})]
Combining Lemmas~\ref{lem:phase1} and~\ref{lem:phase2}, it follows that
there can be at most $\ell = k(2\log{(n/2)} + 1) + O((k/\alpha)\log^2{n}\log{\Delta}) = O((k/\alpha)\log^2{n}\log{\Delta})$ successful
upgrade and fill phases before all $k$ tokens are spread.

By Lemma~\ref{lem:progress},
if the token spreading is not yet complete,
then the probability that the current phase is successful is lower bounded by some constant probability $p>0$.
The actual probability might depend on the execution history up until the current phase,
but the lower bound of $p$ always holds, regardless of this history. 
We can tame these dependencies with a stochastic dominance argument.

In more detail, for each phase $i$ before the tokens are spread, we define a trivial random variable $\hat X_i$ that is $1$ with independent
probability $p$, and otherwise $0$.
Let $X_i$, by contrast, be the random indicator variable that is $1$ if phase $i$ is successful, and otherwise $0$.
For each phase $i$ that occurs after the tokens are spread, $\hat X_i = X_i = 1$ by default.

Note that for each $i$, $X_i$ stochastically dominates $\hat X_i$.
It follows that if $\hat Y_T = \sum_{i=1}^T \hat X_i$ is greater than some $x$ with some probability $\hat p$,
then $Y_T = \sum_{i=1}^T X_i$ is greater than $x$ with probability at least $\hat p$.

With this established, consider the first $T = (c/p)\ell$ phases, for some constant $c\geq 2$.
Note that for this value of $T$, $E[\hat Y_T] = c\ell$.
Because $\hat Y_T$ is the sum of independent random variables,
we concentrate around this expectation. 
In particular, we once again apply the following form of a Chernoff Bound:

\[ \Pr(Y \leq (1-\delta)\mu) \leq e^{- \frac{\delta^2\mu}{2}}, \]

\noindent for $Y=\hat Y_T$, $\delta= 1/2$, and $\mu = c\ell$,
to derive that the probability that $\hat Y_T \leq (c/2)\ell \geq \ell$,
is upper bounded by $e^{-\frac{c\ell}{8}}$.
The same bound therefore holds for the probability that $Y_T \leq (c/2)\ell$.
Notice that this error bound is polynomially small in $n$ with an exponent that grows with constant $c$.
It follows, therefore, that with high probability in $n$,
that token spreading succeeds in the first
$T =  (c/p)\ell = O( (k/\alpha)\log^2{n}\log{\Delta} )$ phases.

To achieve the final {\em round} complexity bound claimed by the theorem statement,
we multiply this upper bound on phases by the length of $\log{N}$ rounds per phase.
\end{proof}

\section{Random Gossip in the Asynchronous Mobile Telephone Model}
\label{sec:amtm}

The mobile telephone model captures the basic dynamics of the peer-to-peer libraries included in standard smartphone operating systems.
This abstraction, however, makes simplifying assumptions---namely, the assumption of synchronized rounds.
In this section we analyze the performance of simple random gossip processes in a more realistic version of the model that eliminates the synchronous round assumption.
In particular, we first define the  asynchronous mobile telephone model (aMTM), which describes an event-driven
peer-to-peer abstraction in which an adversarial scheduler controls the timing of key events in the execution.

An algorithm specified in the aMTM should be directly implementable on real hardware without the need to synchronize or simulate rounds.
This significantly closes the gap between theory and practice. 
With this in mind, after defining the aMTM, we specify and analyze a basic random gossip process strategy.
In this more realistic asynchronous model, different processes can be running at vasty different and changing speeds,
invalidating the clean round-based analysis from the previous section.
We will show, however, that even in this more difficult setting, 
random gossip processes can still be analyzed and shown to spread tokens with speed that increases with available connectivity.

\subsection{The Asynchronous Mobile Telephone Model}

Since the pattern of communication in the asynchronous setting can be complex, our first goal in creating our new abstraction is to impose a simple but flexible structure for how processes communicate with each other. To this end, we introduce a meta-algorithm that is run by each process individually, independent of all others processes in the network. This allows us to analyze the running time of a particular instance of an algorithm and, from there, the performance of the algorithm across all concurrent network instances.

We will require two primary properties from our algorithmic structure. First, for our protocols to be truly asynchronous, they will not be able to follow a static procedural flow. Namely, after perfoming some action, an algorithm in this model may have to wait an indeterminate amount of time before performing another action or even being notified of the results of the first action. While we can parameterize an upper bound for this delay in the model for the sake of our analysis, it is unrealistic for an instance of the algorithm to be aware of this parameter. Second, we would like to abstract away the details of the asynchronous communication from the specfics of the algorithm, allowing us to keep our algorithm descriptions as simple as possible.
\begin{algorithm}
\caption{The Asynchronous MTM Interface}
\label{alg:amtm}
\begin{algorithmic}[1]
\State
\State \underline{Initialization:}
\State
\State $neighbors \gets [:]$
\State $state \gets$ {\tt idle}
\State $receiver \gets$ {\tt null}
\State \textproc{Initialize()}
\State
\While{{\tt true}}
\State
\State $tag \gets$ \textproc{GetTag}()
\State {\tt update}($tag$)
\State
\State $neighbors\gets$ {\tt blockForNeighborUpdates}()
\State
\State $receiver \gets$ \textproc{Select}($neighbors$)
\If{$receiver\neq {\tt null}$}
		\State $state \gets$ {\tt blockForConnection}($receiver$)
		\EndIf
\State
\If{$state$ = {\tt connected}}
\State \textproc{Communicate}($receiver$)
\State $state \gets$ {\tt idle}
\EndIf
\EndWhile
\end{algorithmic}

\end{algorithm}

We accomplish both of these goals by implementing a structure that resembles a looped synchronous algorithm but regulates its execution through access to data members that are updated asynchronously. Formalized in Algorithm \ref{alg:amtm}, the protocol initializes three fields:
\begin{itemize}
\item $neighbors$: A key-value store of references to neighboring processes whose advertisements have been received along with their advertisement tags. This set is maintained asynchronously by the model and updated whenever a new advertisement is received. Whenever a new advertisement is received, it replaces the last known advertisement for the corresponding neighboring process.
\item $state$: An enumerated type field chosen from the set $\{{\tt idle}, {\tt connected}\}$. Also modified asynchronously by the model, this field signifies the current progress in any connections the process is involved in.
\item $receiver$: A nullable reference to a single neighbor for communication purposes after a connection is formed.
\end{itemize}
While these fields accomplish our first goal of enabling our algorithms to execute asynchronously, we satisfy our second goal of abstracting communication details from the implementing algorithm by exposing an interface of four functions:
\begin{itemize}
\item \textproc{Initialize}(): Initialization of algorithm-specific data.
\item \textproc{GetTag}(): Return the advertisement tag for this process which is then broadcast to all neighboring processes.
\item \textproc{Select}($neighbors$): Return a neighbor (or {\tt null} for no neighbor) to connect to from among those discovered.
\item \textproc{Communicate}($receiver$): Perform a bounded amount of communication with selected neighbor $receiver$.
\end{itemize}
The execution of an iteration of the algorithm loop begins by getting the process' advertisement tag and broadcasting it to all neighboring processes. The model then blocks until a reference to a neighboring process is added to the $neighbors$ set. Once the $neighbors$ set contains at least one neighbor, the implementing algorithm selects one neighbor from the set and returns it. If the selected neighbor isn't {\tt null}, the protocol then attempts to connect with the selected neighbor, and blocks for another indeterminate duration of time for the connection attempt to succeed $(state\gets {\tt connected})$ or fail $(state\gets {\tt idle})$. If the connection succeeds, the two connected processes communicate before proceeding to the next iteration.

We assume that each step of the protocol executes instantly with the exception of the model functions {\tt blockForNeighborUpdates}() and {\tt blockForConnection()} and the algorithm function \textproc{Communicate}($receiver$). These functions implicitly block the protocol's execution. The model functions block the execution until the $neighbors$ and $state$ fields are available to be referenced by the algorithm, respectively, while \textproc{Communicate}($receiver$) stalls until the connected nodes communicate. In order for this abstraction to be useful to our analysis, however, we need to parameterize the maximum duration of these blocking events. We therefore define the corresponding model parameters $\delta_{update}$, $\delta_{connect}$, $\delta_{comm}$, and $\delta_{old}$ which are not known in advance and can change between executions:
\begin{itemize}
\item $\delta_{update}$: If a process $u$ calls {\tt update}($tag$) at time $\delta$, $u$ will be added to the $neighbors$ set of all neighboring processes by time $\delta + \delta_{update}$ at the latest. This is the maximum time for step 14 of the protocol.
\item $\delta_{old}$: Conversly, if a process $u$ calls {\tt update}($tag$) at time $\delta$, no neighboring process will add $u$ to their neighbors set after time $\delta_{old}$ where $\delta_{old} > \delta_{update}$.
\item $\delta_{connect}$: If a process $u$ calls {\tt connect}($v$) at time $\delta$, by time $\delta + \delta_{connect}$ at the latest, either the connection attempt will have failed or $u$ and $v$ will have succesfully connected. This is the maximum time for step 18 of the protocol.
\item $\delta_{comm}$: As stated in the model description, once a connection is formed, the connected processes may engage in a bounded amount of communication, $\delta_{comm}$ defines the maximum time required for this communication to occur. This is the maximum time for step 21 of the protocol.
\end{itemize}

Notice that the specified model only defines how to attempt outgoing connections. While this abstraction is similar to the mobile telephone model in that it restricts a process to one such connection attempt at a time, it will deviate slightly by allowing a single incoming connection attempt as well. This allowance will ease our analysis of algorithms in this setting as it frees a process to accept an incoming connection attempt regardless of its current state. For now, we will assume the process of accepting incoming connection attempts is simply to accept the first connection attempt received and call \textproc{Communicate}($sender$) where $sender$ is the source of the incoming connection.

\subsection{The Asynchronous Random Spread Gossip Algorithm}

We now instantiate our algorithm as a particular instance of the asynchronous mobile telephone model protocol by implementing the four functions specified by the interface. First we initialize the token set of the process to contain any tokens it knows. We also instantiate the hash function used for creating the advertisement tags:\\
\begin{algorithmic}
\Function{Initialize}{}
\State $tokens \gets$ initial tokens (if any) known by $u$
\State $H \gets$ a hash function
\EndFunction
\end{algorithmic}
\bigskip
Next we define the tag function to simply return the hash of the token set that the process knows:\\
\begin{algorithmic}
\Function{GetTag}{}
\State
\Return $H(tokens)$
\EndFunction
\end{algorithmic}
\bigskip
To select a neighbor from those that a process has discovered, the algorithm will first create a filtered set of neighbors to only include those that would be productive to connect to (those neighbors with different token hashes). Then, following the random gossip strategy, it will select one such neighbor uniformly at random. If no productive neighbor exists then the algorithm doesn't select any neighbor and remains idle. Lastly, note that when a productive neighbor is selected, the algorithm clears its set of known neighbors. As we will see in Lemma \ref{lem:fault}, refreshing the set of known nearby processes minimizes the effect of faulty nodes on performance.\\
\begin{algorithmic}
\Function{Select}{neighbors}
\State $productiveNeighbors\gets\emptyset$
\State
\For{$neighbor$ in $neighbors$}
\If{$neighbor.value \neq \textproc{GetTag}()$}
\State $productiveNeighbors.add(neighbor.key)$
\EndIf
\EndFor
\State
\If{$productiveNeighbors \neq \emptyset$}
\State $neighbors \gets [:]$ // remove stale advertisements
\State \text{// chosen uniformly at random}
\State
\Return $receiver \in productiveNeighbors$
\Else \State
\Return {\tt null}
\EndIf
\EndFunction
\end{algorithmic}
\bigskip
Finally, if two processes form a succesful connection, they exchange a single token in the symmetric set difference between their two token sets:\\
\begin{algorithmic}
\Function{Communicate}{$receiver$}
\State $t\gets$ some $t\in (T$ $\Delta$ $receiver.T$)
\State (exchange token $t$)
\EndFunction
\end{algorithmic}

\subsection{Asynchronous Random Spread Gossip Analysis}
In this section we analyze the above algorithm. We begin with a proof of convergence, showing that in the worst case the asynchronous random spread gossip algorithm spreads all tokens to all nodes in the network in time $\bigO{nk\delta_{max}}$. We then take advantage of the vertex expansion $\alpha$ to demonstrate how it increases the rate at which a single token is spread.

\subsubsection{Proof of Convergence}
We begin our analysis by showing that the asynchronous random spread algorithm spreads all $k$ tokens to the entire network in time at most $\bigO{nk\delta_{max}}$. Firstly, for our analysis of the asynchronous setting, we will have to redefine our notion of the productive subgraph.

\begin{definition}
At time $\delta$, define $G_\delta$ to be the productive subgraph of the network $G=(V,E)$ at this time such that $G_\delta=(V, E_\delta)$ where $E_\delta = \{(u,v): H(u.tokens) \neq H(v.tokens)$ at time $\delta\}$.
\end{definition}

Notice, as in the previous section, we assume the very low probability event of hash collisions do not occur.
That is: $H(u.tokens) = H(v.tokens) \Longleftrightarrow u.tokens = v.tokens$.
With this in mind, we establish our first bound (remember in the following that $\delta_{update}$, $\delta_{connect}$, and $\delta_{comm}$ are the relevant
maximum time bounds---unknown to the algorithm---for key model behavior).

\begin{lemma}
The asynchronous random gossip algorithm takes time $O(nk\delta_{max})$ to spread all tokens where $n$ is the number of nodes in the network, $k$ is the number of tokens to spread, and $\delta_{max}=O(\delta_{update} + \delta_{connect}+\delta_{comm})$ is the maximum amount of time between iterations of the algorithm loop.
\end{lemma}
\begin{proof}
Fix some time $\delta$. Our goal is to show that within the interval $\delta$ to $\delta+\delta_{max}$,
at least one node learns a new token.
Because
this can only occur at most $nk$ times before all nodes know all tokens,
if we can show the above we have established the lemma.

Fix some time $\delta$. Let $G_{\delta}$ be the {\em productive subgraph} (see the above definition) at the beginning of this interval.
  If not all tokens have spread, clearly there exists a node $u$ such that the $deg(u) > 0$ in $G_{\delta}$. 
  
  By the guarantees of the model, 
  by time $\delta' \leq \delta + \delta_{update} + \delta_{connect}+\delta_{comm}$ , $u$ will have heard advertisements from all neighbors in $G_{\delta}$,
  and then subsequently looped back to the top of its main connect loop.
  
  For each neighbor $v$ in $G_{\delta}$, either $u$ adds $v$ to its set, or at some point after $\delta$, $v$ and $u$'s token sets changed
  such that $u.tokens = v.tokens$, preventing $u$ from adding $v$.
  In this case, however, at least one new token was learned by some node and we are done.
  If this is not the case, then $u$ now has a non-empty $productiveNeighbors$ set.
  
  Going forward, let $v$ be the node $u$ randomly chooses from this set. If the connection fails, this indicates that $v$ is involved in another connection with some other node $v'$.
  If the connection is successful, then $u$ and $v$ will exchange a token.
  Either way, a new token is learned by some node in $\{u,v,v'\}$
  in at most another $\delta_{connect}$ time.
  
  The total amount of time for some node to learn something new is in $O(\delta_{update}+\delta_{connect}+\delta_{comm})$, as needed.

\end{proof}

\begin{lemma} \label{lem:fault}
Let $t$ be the maximum number of faulty nodes in the network, the asynchronous random gossip algorithm takes time $O(\delta_{max}(nk+t))$ to spread all tokens.
\end{lemma}
\begin{proof}

Again, consider the productive subgraph $G_{\delta}$ at a particular time $\delta$ for a node $u$ when its $neighbors$ set is empty. If no nodes leave the subgraph then $u$ is guaranteed to learn of all these neighbors and add them to its $neighbors$ set. However, now allow some node $v$ in $u$'s $neighbors$ set to experience a failure between times $\delta-\delta_{old}$ and $\delta+\delta_{max}$ (if the failure happens before $\delta-\delta_{old}$ then by the guarantee of the aMTM, $u$ will not have received $v$'s update). Upon entering an iteration of the outer loop, $u$ may attempt to connect with $v$ since $v$'s advertisement is still fresh. In this event, which is clearly the worst case, the connection fails and time at most $\delta_{max}$ was spent since this is the maximum amount of time the outer loop can possibly take.
\par This failure can happen in each new iteration of the outer loop for at most time $\delta_{old}$, at which point the advertisement ceases to update $u$'s neighbor set. Therefore, a single failed node can cause a delay of time at most $\delta_{old} + 2\delta_{max}$. Since there are $t$ faulty nodes, this introduces a total slowdown of $t(\delta_{old} + 2\delta_{max})$. Therefore, the time for this algorithm to spread all $k$ tokens is $\bigO{nk\delta_{max}} + t(\delta_{old} + \delta_{max})$. Furthermore, if we assume $\delta_{old} = \bigO{\delta_{update}}=\bigO{\delta_{max}}$, $\bigO{nk\delta_{max} + t2\delta_{max}}=\bigO{\delta_{max}(nk+t)}$.
\end{proof}
\begin{lemma}
Let $b$ be the maximum fraction of neighbors for a node $u$ that can be byzantine, the asynchronous random gossip algorithm takes time $O(nk\delta_{max}/(1-b))$ in expectation to spread all tokens.
\end{lemma}
\begin{proof}
If the productive subgraph stays connected, the worst event that can occur during the interval of length $\delta_{max}$ is that an honest node chooses a byzantine neighbor to connect to. This happens with probability at most $b$ and therefore a node engages in a productive, honest connection with probability at least $1-b$. Consider the series $m$ of intervals of time at most $\delta_{max}$ and label them with the indicator variables $X_1,\ldots,X_m$ such that:
\[
X_i =
\begin{cases}
0 &\text{if the node in interval $i$ connects to a}\\ &\text{byzantine node} \\
1 &\text{otherwise }
\end{cases}
\]
\begin{align*}
nk &= E\big[\sum_{i=1}^{i=m}{X_i}\big] \\
&= \sum_{i=1}^{i=m}{E[X_i]} \\
&= \sum_{i=1}^{i=m}{1-b} = m(1-b)
\end{align*}
Therefore, achieving $nk$ successes in expectation, would take $m=\frac{nk}{1-b}$ intervals. Since each interval takes at most $\delta_{max}$ time, the algorithm takes time $O(nk\delta_{max}/(1-b))$.
\end{proof}

\subsubsection{Analysis of Spreading a Single Token}

We now analyze the spread of a single token in the network to demonstrate that the performance of the algorithm still improves with the vertex expansion of the network $\alpha$ in an asynchronous setting. Our goal in this subsection is to prove the following time bound to spread a single token:

\begin{theorem} \label{thm:async}
The asynchronous random spread gossip algorithm takes time at most $\bigO{\delta_{max}\sqrt{n/\alpha}\log^2{(n\alpha)}}$, where $n$ is the number of nodes in the network, $\alpha$ is the vertex expansion, and $\delta_{max}$ is the maximum time required for an iteration of the asynchronous mobile telephone model loop.
\end{theorem}

 Unlike with our analysis of the synchronous algorithm, we cannot directly leverage a productive subgraph that remains stable through synchronized rounds.
We must instead identify cores of useful edges amidst the unpredictable churn and argue that over a sufficiently long interval they deliver a sufficiently large number of new tokens.

%
%

We accomplish this by fixing the productive subgraph at $G_{\delta}$ and observe an interval of length $2\delta_{max}$. During this interval, we want to show that for every edge $(u,v)\in E_{\delta}$ such that $u$ is informed and $v$ is uninformed, either $v$ becomes otherwise informed or $u$ returns $v$ from \textproc{Select}($neighbors$) with good probability during this interval. Namely, this probability is lower-bounded by the probability $u$ would return $v$ if $neighbors$ included all of $u$'s neighbors from $G_{\delta}$ itself.

\begin{lemma} \label{lem:randomlemma}
For a fixed time $\delta$ and fixed edge $(u,v)\in E_{\delta}$ such that $u$ knows the token and $v$ does not, if $v$ does not otherwise learn the token in this interval, node $u$ returns $v$ from \textproc{Select}($neighbors$) uniformly at random from a set of at most $deg(u)$ nodes where $deg(u)$ is the degree of $u$ in the productive subgraph and the resulting connection attempt concludes no later than time $\delta+2\delta_{max}$.
\end{lemma}

\begin{proof}
Fix the productive subgraph at this time, $G_{\delta}$ and fix an informed node $u$ and uninformed node $v$. Since $u$ is an informed node, all of its edges in the productive subgraph are incident to uninformed nodes. Since nodes never forget the token, the number of uninformed nodes can only decrease. Now consider an execution of \textproc{Select}($neighbors$) before time $\delta + \delta_{max}$ in which $v$ is added to $u.productiveNeighbors$. Since by assumption $v$ does not otherwise learn the token in this interval, it must be the case that $v$ advertised its uninformed status in this interval and been included in $u.neighbors$ and subsequently $u.productiveNeighbors$ so we know this occurs at least once in the interval $\delta + \delta_{max}$ (the extra time $\delta_{connect}+\delta_{comm}$ is to allow an additional iteration of $u$'s loop before \textproc{Select} is called). Furthermore, we know that since the number of uninformed neighbors can't increase from that in the productive subgraph $G_\delta$, there can be at most $deg(u)$ neighbors in $u.productiveNeighbors$. Since $u$ returns a particular neighbor from this set with uniform randomness, the probability that $u$ returns $v$ is at least $1/deg(u)$. Furthermore, regardless whether or not the resulting connection attempt is a success or a failure, it finishes in at most $\delta_{connect}+\delta_{comm}$ additional time for a total maximum time of $\delta_{update}+2(\delta_{connect}+\delta_{comm})<2\delta_{max}$.

\end{proof}
Now that we have quantified the amount of time necessary for a node to successfully connect, we need an estimate for how many connections we can expect to be succesful. Similar to our previous analysis, this is dependent on the amount of competition between connection attempts sent to a single node.
We begin with a useful graph theory definition.
\begin{definition}
For a graph $G=(V,E)$, we define the \textbf{degree weight} of a node be the sum of the weights of all incoming edges, where the weight of each edge $(u,v)$ is $1/deg(u)$. Formally:
\[w(v)=\sum_{\forall u\in V, (u,v)\in E}w(u,v)=\sum_{\forall u\in V, (u,v)\in E}1/deg(u)\]
\end{definition}

We now prove a useful result about one-round random matchings in a bipartite graph that leverages our degree weight definition in its proof.

\begin{lemma} \label{lem:sqrtlemma}
For a bipartite graph G=(X, Y, E) with edge independence $\edgeindnum{B(X)}=|X| = m$. 
Assume each node $u\in X$ selects a neighbor with uniform randomness with probability $1/deg(u)$. With at least constant probability, at least $\sqrt{m}/\log{m}$ distinct nodes from $Y$ are selected.
\end{lemma}
\begin{proof}
Partition the nodes of $Y$ into a ``core" set of nodes $Z$ with constant degree weight, and a ``non-core" set of nodes $Y \setminus Z$ with less than constant weight. 

We first consider the case where $|Z| \geq \sqrt{m}$. Here it is sufficient to show that nodes with at least constant weight are selected with constant probability. 
For a node $v\in Z$ such that $(u,v)\in E$, the probability that $u$ does not select $v$ is at most $1-1/deg(u)=1-w(u,v)$. Therefore the probability that $v$ is selected by {\em some} node is:
\begin{align*}
\Pr[\text{$v$ is selected}] &\geq 1 - \Pi_{u,(u,v)\in E}(1 - w(u,v))\\
&\geq 1 - \Pi_{u, (u,v)\in E}e^{-w(u,v)}\\
&\geq 1 - e^{-\sum_{u, (u,v)\in E}w(u,v)}\\
&\geq 1 - e^{-w(v)}
\end{align*}
Since by our assumption $w(v)$ is a constant, $v$ is selected with at least constant probability. If we denote this probability $p$, we can express the probability that $v$ is not selected as $1-p$. Therefore, the expected number of nodes in the core set that are not selected is at most $(1-p)\sqrt{m}$. Let $W$ be the number of core nodes that are not selected, we can apply Markov's inequality to demonstrate that the probability we exceed this expectation by more than a constant fraction is at most constant:
\begin{align*}
\Pr[W\geq 2(1-p)\sqrt{m}] \leq \frac{1}{2}
\end{align*}
Therefore, with at least a constant probability, $\bigO{\sqrt{m}}$ nodes are selected from the core set.



Now consider the case where $|Z|<\sqrt{m}$. Observe that for a node $u\in X$ that neighbors a node in $Y\setminus Z$, the sum of the edge weights for edges $(u,v)$ such that $v\in Y\setminus Z$ is at least $1/\sqrt{m}$. This is because $u$ can select at most $\sqrt{m}-1$ other nodes that are not in $Y\setminus Z$. Therefore, for each node in $X$ that neighbors a node in $Y\setminus Z$, the node in $X$ chooses a node in $Y\setminus Z$ with probability at least $1/\sqrt{m}$. Since there must be $m-|Z|$ such nodes in $X$ that neighbor nodes in $Y\setminus Z$, in expectation at least $\bigO{m/\sqrt{m}}=\bigO{\sqrt{m}}$ nodes in $X$ select a node in $Y \setminus Z$.\\

Next, conditioned on the event that $\sqrt{m}$ nodes from $X$ select non-core nodes, we need to show that not too many of the nodes in $Y \setminus Z$ are chosen multiple times. Namely, we would like to show that the probability that any node is selected by more than $c\log{m}$ nodes (for some sufficiently large constant $c$) from $X$ that choose a non-core node is small. Fix a node $v\in Y$ and define the indicator variable $I_j$ as follows:
\[
X_j =
\begin{cases}
1\text{ the $j$th node in $X$ selects $v$}\\
0\text{ otherwise}
\end{cases}
\]
Since the size of the maximum matching is size $m$ and there are $m$ nodes in $Y$, we know that $v$ has at most constant degree weight and therefore, in expectation, is selected by at most a constant number of nodes from $X$. Denote this constant expectation $\mu$ and apply the following Chernoff bound to the sequence $I_1,\ldots,I_m$ with expectation $\mu$ to upper bound the probability that the total number of such nodes $I=\sum_{j=1}^{m}I_j$ exceeds $c\log{m}$. For a sufficiently large constant $c$ and constant $\mu$ we find that this probability is polynomially-small in $m$:
\begin{align*}
\Pr[I\geq (1+ \epsilon)\mu] &\leq e^{-\frac{\epsilon\mu}{3}} \\
\Pr[I\geq \mu c \log{m}] &\leq e^{-\frac{\mu(c\log{m}-1)}{3}} = \frac{e^{\mu/3}}{me^{c\mu/3}} \leq \frac{1}{m}
\end{align*}
Therefore, if we apply the union bound over the at most $m-\sqrt{m}\leq m$ nodes in $Y\setminus Z$ we can upper bound the probability that any such node is selected by at least $c\log{n}$ nodes in $X$:
\begin{align*}
  \sum_{i\in[m]}\frac{e^{\mu/3}}{me^{c\mu/3}}=m\frac{e^{\mu/3}}{me^{c\mu/3}} < constant
\end{align*}
Therefore, with at least constant probability no node is selected by at least $c\log{n}$ nodes from $X$ that choose non-core nodes. Therefore, given that $\sqrt{m}$ nodes select non-core nodes, with at least a constant probability at least $\bigO{\sqrt{m}/\log{m}}$ nodes are selected.
\end{proof}

The above lemmas allow us to quantify the number of successful connections made in an interval of length $\delta_{update}+2\delta_{max}$ with respect to $G_\delta$ for some time $\delta$, but we need to relate this result back to the productive subgraph as a whole.
\begin{lemma} \label{lem:step}
Fix the productive subgraph $G_\delta$ with a maximum matching of size $m$. With high probability, $\sqrt{m}/\log{m}$ successful connections will occur by time $\delta+ \delta_{update}+2\delta_{max}$.
\end{lemma}

\begin{proof}
For a fixed $G_{\delta}$, if we consider an edge $(u,v)\in E_{\delta}$ consisting of an informed node $u$ and uninformed node $v$ such that $v$ is not otherwise informed in this interval, we know from Lemma \ref{lem:randomlemma} that $u$ adds $v$ to $u.productiveNeighbors$ some time before $\delta+\delta_{update}+\delta_{max}$. Since $u$ returns $v$ from \textproc{Select}($neighbors$) with probability at least $1/deg(u)$, inclusion of $v$ in $u.produtveNeighbors$ represents a selection weight of at least $1/deg(u)$ for the edge $(u,v)$. Since the selection weight for this edge never decreases, the edge weight for $(u,v)$ accumulated in this interval (and therefore its selection probability) must be at least $1/deg(u)$ (the weight the edge would have in the productive subgraph itself). Therefore, the collection of these edge weights observed over this interval represents a bipartite graph with maximum matching of size $m$ where each edge is selected with probability at least $1/deg(u)$. Therefore, according to Lemma \ref{lem:sqrtlemma}, with at least constant probability, $\sqrt{m}/\log{m}$ nodes are selected over this interval. Furthermore, each connection attempt takes at most $\delta_{max}$ time which concludes our Lemma as long as Lemma \ref{lem:randomlemma} holds.

However, since Lemma \ref{lem:randomlemma} assumes that for each of these productive edges $(u,v)$, $v$ is not otherwise informed, we must consider this case as well. However, since $(u,v)\in E_{\delta}$ and $v$ can only have been informed through a prior successful connection during this interval, it should be clear that this event does not reduce the number of successful connections that take place during this interval and so the Lemma is still satisfied.
\end{proof}
To continue our analysis with respect to the vertex expansion, we now relate the expected number of productive connections at fixed points in time to the $\alpha$.
\begin{lemma} \label{lem:constantfraction}
Let $S(\delta)$ be the subset of informed nodes such that $n(\delta)=|S(\delta)|$ and $n^*_\delta=\min(|S(\delta)|, |V\setminus S(\delta)|)$. Furthermore, abbreviate $\delta_{update}+2\delta_{max}$ to $\delta_{interval}$. With high probability, if $n(\delta) \leq n/2$, it takes at most time $2\delta_{interval}\log{(n(\delta)\alpha)}\sqrt{n(\delta)/\alpha}$ to at least double the number of informed nodes. Explicitly, with high probability: $$n(\delta + 2\delta_{interval}\log{(n(\delta)\alpha)}\sqrt{n(\delta)/\alpha})\geq 2n(\delta)$$
\end{lemma}

\begin{proof}
Consider a sequence of fixed times $\delta_0,\ldots,\delta_t$ that are time $\delta_{interval}$ apart. Lemma \ref{lem:step} estimates the number of succesful connections with respect to the size of the maximum matching while Lemma \ref{lem:msize} which relates the size of the maximum matching size $m$ to the vertex expansion such that $m\geq n(\delta)\alpha$ for $0<n(\delta)\leq n/2$. Therefore, with high probability, the number of nodes that become informed between $\delta_{i-1}$ and $\delta_i$ is $n(\delta_{i-1}) + (1/2)\sqrt{n(\delta_{i-1})\alpha}/\log{(n(\delta_{i-1})\alpha)}$. Therefore, the number of nodes that are informed by time $\delta_t$ is: $$n(\delta_t)=n(\delta_0)+(1/2)\sqrt{n(\delta_{0})\alpha}/\log{(n(\delta_{0})\alpha)}+\ldots$$
$$\ldots+(1/2)\sqrt{n(\delta_{t-1})\alpha}/\log{(n(\delta_{t-1})\alpha)}$$
Since clearly $n(\delta_j)\geq n(\delta_{i})$ for any $j\geq i$, we can simplify the above:
$$n(\delta_t)\geq n(\delta_0) + (t/2)\sqrt{n(\delta_{0})\alpha}/\log{(n(\delta_{0})\alpha)}$$
Lastly, if we set $n(\delta_t)=2n(\delta_0)$ we can solve for $t$:
\begin{align*}
2n(\delta_0) &\geq n(\delta_0) + (t/2)\sqrt{n(\delta_{0})\alpha}/\log{(n(\delta_{0})\alpha)}\\
n(\delta_0) &\geq  (t/2)\sqrt{n(\delta_{0})\alpha}/\log{(n(\delta_{0})\alpha)}
\end{align*}
\[2\log{(n(\delta_{0})\alpha)}\sqrt{n(\delta_{0})/\alpha} \geq  t \]

Since there are most $2\log{(n(\delta_{0})\alpha)}\sqrt{n(\delta_{0})/\alpha}$ steps of length $\delta_{interval}$, the total time to double the number of informed nodes from $n(\delta_0)$ is at most $2\delta_{interval}\log{(n(\delta_{0})\alpha)}\sqrt{n(\delta_{0})/\alpha}$.
\end{proof}
We now use the length of this interval to analyze the time required to spread the token to half of the nodes in the network.
\begin{lemma} \label{lem:half}
It takes time at most $\bigO{\delta_{max}\sqrt{n/\alpha}\log^2{(n\alpha)}}$ to spread the token to $n/2$ nodes in the network.
\end{lemma}
\begin{proof}
By Lemma \ref{lem:constantfraction}, we can see that if there $n(\delta)$ informed nodes for a given time $\delta$, after time $2\delta_{interval}\log{(n(\delta)\alpha)}\sqrt{n(\delta)/\alpha}$, with high probability we at least double the number of informed nodes. Therefore, to find the number $t$ of intervals required, it suffices to solve for $T$ such that:
\[2^{T-1}=n/2\]
Which yields $T=\log{n}$. Therefore, since each interval takes time at most $\bigO{\delta_{interval}\sqrt{n/\alpha}\log{(n\alpha)}}=\bigO{\delta_{max}\sqrt{n/\alpha}\log{(n\alpha)}}$, the total time required is $\bigO{\delta_{max}\sqrt{n/\alpha}\log^2{(n\alpha)}}$.
\end{proof}
We now have all the necessary components to prove our main theorem about the time required to spread the token to all nodes in the network.
\begin{proof} [Proof (of Theorem~\ref{thm:async})]
When $n(\delta) > n/2$, our goal is to reduce the number of uninformed nodes by half. However, we can no longer relate the size of the maximum matching to the number of informed nodes since we are instead limited by the uninformed nodes since $|V\setminus S_\delta| < |S_\delta|$. Therefore, for fixed times $\delta_0,\ldots,\delta_t$ which are $\delta_{interval}$ apart, we can express the number of uniformed nodes at time $\delta_t$, $n^*({\delta_t})$, as:
$$n^*(\delta_t)=n^*(\delta_0)-(1/2)\sqrt{n^*(\delta_{0})\alpha}/\log{(n^*(\delta_{0})\alpha)}-\ldots$$
$$\ldots-(1/2)\sqrt{n^*(\delta_{t-1})\alpha}/\log{(n^*(\delta_{t-1})\alpha)}$$
Setting $n^*({\delta_t})=n^*({\delta_0})/2$ and solving for $t$ yields:
$$n^*(\delta_0)/2 =n^*(\delta_0)-(1/2)\sqrt{n^*(\delta_{0})\alpha}\log{(n^*(\delta_{0})\alpha)}-\ldots$$
$$\ldots-(1/2)\sqrt{n^*(\delta_{t-1})\alpha}/\log{(n^*(\delta_{t-1})\alpha)}$$
$$\geq (t/2)\sqrt{n^*(\delta_{t-1})\alpha}/\log{(n^*(\delta_{t-1})\alpha)}$$

Which shows that $t\leq 2\sqrt{n(\delta_{t-1})/\alpha}\log{(n(\delta_{t-1})/\alpha)}$, similar to Lemma \ref{lem:constantfraction}. We can apply a proof symmetric to that of Lemma \ref{lem:half} to relate this time to the time to spread the token to all the remaining $n/2$ nodes. Observe that again we need $T=\log{n}$ intervals of length $2\delta_{max}\sqrt{n^*(\delta)/\alpha}\log{(n^*(\delta)\alpha)}$ since we are halving the number of uniformed nodes each time. Therefore, the total time once again is $\bigO{\delta_{max}\sqrt{n/\alpha}\log^2{(n\alpha)}}$ to spread the token to all remaining nodes. Therefore, the total running time of the algorithm is $\bigO{\delta_{max}\sqrt{n/\alpha}\log^2{(n\alpha)}}$.
\end{proof}

\bibliographystyle{IEEEtran}
\bibliography{randomGossip}

\end{document}